\documentclass[final]{IEEEtran}
\pdfoutput=1
\usepackage{amsthm,amssymb,graphicx,multirow,amsmath,color,amsfonts}
\usepackage[update,prepend]{epstopdf}
\usepackage[noadjust]{cite}
\usepackage[latin1]{inputenc}

\def\nb0{{\mathbf{0}}}
\def\nb1{{\mathbf{1}}}

\def\ncalL{{\mathcal{L}}}

\def\nrmd{{\rm d}}

\newtheorem{lemma}{Lemma}

\newtheorem{theorem}{Theorem}
\newtheorem{prop}{Proposition}
\newtheorem{cor}{Corollary}

\def\E{\mathbb{E}}

\def\P{\mathbb{P}}
\def\pc{\mathtt{P_c}}

\def\ie{{\em i.e.}}

\def\R{\mathbb{R}}

\def\N{\sigma^2}
\def\T{\beta}							
\def\sinr{\mathtt{SINR}}			
\def\snr{\mathtt{SNR}}
\def\sir{\mathtt{SIR}}

\def\calA{\mathcal{A}}
\def\calK{\mathcal{K}}

\allowdisplaybreaks 
\begin{document}
\graphicspath{{./Figures/}}
\title{Modeling and Analysis of $K$-Tier Downlink Heterogeneous Cellular Networks}
\author{Harpreet S. Dhillon,~\IEEEmembership{Student Member, IEEE}, Radha Krishna Ganti,~\IEEEmembership{Member, IEEE}, Fran\c{c}ois Baccelli and Jeffrey G. Andrews,~\IEEEmembership{Senior Member, IEEE}%
\thanks{Manuscript received March 7, 2011; accepted July 3, 2011. This work was supported by NSF grant CIF-1016649. A part of this paper was presented at  ITA 2011 in San Diego, CA~\cite{DhiGanC2011}.}%
\thanks{
H. S. Dhillon,  and J. G. Andrews are with WNCG, the University of Texas at Austin, USA. Email: dhillon@utexas.edu,  jandrews@ece.utexas.edu.  F. Baccelli is with Ecole Normale Superieure (ENS) and INRIA in Paris, France. Email: Francois.Baccelli@ens.fr.  R. K. Ganti was with UT Austin and is currently with the Indian Institute of Technology Madras, Chennai, India. Email: rganti@ee.iitm.ac.in.}
}

\maketitle
\begin{abstract}
Cellular networks are in a major transition from a carefully planned set of large tower-mounted base-stations (BSs) to an irregular deployment of heterogeneous infrastructure elements that often additionally includes micro, pico, and femtocells, as well as distributed antennas. In this paper, we develop a tractable, flexible, and accurate model for a downlink heterogeneous cellular network (HCN) consisting of $K$ tiers of randomly located BSs, where each tier may differ in terms of average transmit power, supported data rate and BS density. Assuming a mobile user connects to the strongest candidate BS, the resulting Signal-to-Interference-plus-Noise-Ratio ($\sinr$) is greater than 1 when in coverage, Rayleigh fading, we derive an expression for the probability of coverage (equivalently outage) over the entire network under both open and closed access, which assumes a strikingly simple closed-form in the high $\sinr$ regime and is accurate down to $-4$ dB even under weaker assumptions. For external validation, we compare against an actual LTE network (for tier 1) with the other $K-1$ tiers being modeled as independent Poisson Point Processes. In this case as well, our model is accurate to within 1-2 dB.  We also derive the average rate achieved by a randomly located mobile and the average load on each tier of BSs. One interesting observation for interference-limited open access networks is that at a given $\sinr$, adding more tiers and/or BSs neither increases nor decreases the probability of coverage or outage when all the tiers have the same target-$\sinr$.
\end{abstract}


\section{Introduction}

Mathematical analysis for conventional (1-tier) cellular networks is known to be hard, and so highly simplified system models or complex system level simulations are generally used for analysis and design, respectively.  To make matters worse, cellular networks are becoming increasingly complex due to the deployment of multiple classes of BSs that have distinctly different traits \cite{Qualcomm08, AndClaJ2012}.  For example, a typical 3G or 4G cellular network already has traditional BSs that are long-range and guarantee near-universal coverage; operator-managed picocells \cite{KisGre03,WuMur04} and distributed antennas \cite{Sal87,RohPau03,ZhaAnd08,ChoAnd07} that have a more compact form factor, a smaller coverage area, and are used to increase capacity while eliminating coverage deadzones; and femtocells, which have emerged more recently and are distinguished by their end-user installation in arbitrary locations, very short range, and possibility of having a closed-subscriber group \cite{ChaAndGat08,Picochip07,ChaAnd09}.  This evolution toward heterogeneity will continue to accelerate due to crushing demands for mobile data traffic caused by the proliferation of data-hungry devices and applications.

\begin{figure}[ht!]
\centering
\includegraphics[width=.8\columnwidth]{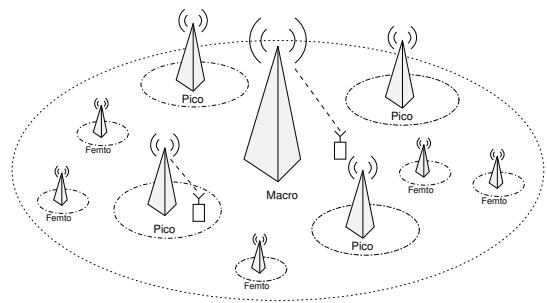}
\caption{Illustration of a three-tier heterogenous network utilizing a mix of macro, pico and femtocell BSs. Only a single macro-cell is shown for the sake of simplicity.
}
\label{fig:HCN}
\end{figure}

A straightforward unifying model for heterogeneous cellular networks (HCNs) would consist of $K$ spatially and spectrally coexisting tiers, where each tier is distinguished by its transmit power, BS density, and data rate, as shown in Fig.~\ref{fig:HCN}.  For example, traditional BSs (tier 1) would typically have a much higher transmit power and lower density and offered rate than the lower tiers (e.g. pico and femtocells).  To visualize what the coverage areas in such a network might look like, consider Figs.~\ref{fig:2HCN_PPP}-\ref{fig:3HCN_RealData}, which show average power-based (equivalently average $\sinr$-based) coverage regions for some plausible two and three tier networks.  Clearly, the coverage, rate, and reliability that mobile users experience in such networks can be expected to be quite different than in traditional cellular networks that use familiar models like the hexagonal grid.

The objective of this paper is to provide a flexible baseline model for HCNs, and to show how it can be used to provide tractable and reasonably accurate analysis of important metrics like the $\sinr$ statistics, outage probability and average rate. Those familiar with cellular network analysis will recognize that this goal is fairly ambitious since such results have been hard to come by even for traditional cellular networks.

\subsection{Related Work and Motivation}

The study and design of conventional (1-tier) cellular networks has often tended towards two extremes.  For analysis and academic research, very simplistic models are typically employed in order to maintain tractability, while for design and development (e.g. in industry) complex system-level simulations with a very large number of parameters are generally used.  This has made it difficult to estimate the actual gain that new techniques developed by researchers might provide in real systems. Well-known examples include multiuser detection \cite{VerBook}, multiuser MIMO \cite{CaiSha03}, and BS cooperation \cite{FosKar06IEE,GesHan10}; all of which promised much larger gains in theory than have been achieved thus far in practice \cite{And05,AnnGor10}.

A popular analytical model for multicell systems is the Wyner model \cite{Wyn94}, which assumes channel gains from all (usually only 1 or 2) interfering BSs are equal and thus constant over the entire cell.  Such a model does not distinguish between cell edge and interior users and in most cases does not even have a notion of outage since $\sinr$ is fixed and deterministic.  It can be tuned to reasonably model average metrics in a system with lots of averaging, such as the CDMA uplink, but is not accurate in general and particularly for systems with 1 or 2 strong interferers, like a typical OFDMA-based $4$G network \cite{XuZhaAnd10}. Another common approach is to consider only a small number of interfering cells (as few as one) and abstract the desired and interfering BSs to an interference channel \cite{Cha11,JinTse07}. Finally, perhaps the most popular and accepted model is the two-dimensional hexagonal grid model. The grid model is frequently used as the basis of system-level simulations but analysis is not generally possible \cite{CatDriGre00,GanKri97,EkiErs01}. However, both the scalability and the accuracy of grid model are questionable in the context of network heterogeneity (see Figs.~\ref{fig:2HCN_PPP}-\ref{fig:3HCN_RealData}).

A less accepted model is to allow the locations of the BSs to be drawn from a stochastic point process \cite{Bro00,BacKle97,BacZuy97}. Such a model seems sensible for femtocells -- which will take up unknown and unplanned positions -- but perhaps dubious for the higher tiers which are centrally planned. Nevertheless, as Figs.~\ref{fig:2HCN_PPP}-\ref{fig:3HCN_RealData} show, the difference between randomly placed and actual planned locations may not be as large as expected, even for tier 1 macro BSs. Indeed, the recent work \cite{AndBacGan10} showed that for a one-tier network, even with the BS locations drawn from a Poisson Point Process (PPP), the resulting network model is about as accurate as the standard grid model, when compared to an actual $4$G network. Importantly, such a model allows useful mathematical tools from stochastic geometry to be brought to bear on the problem \cite{StoKen96, BacBlaNOW, HaeAnd09}, allowing a tractable analytical model that is also accurate. This model has very recently been extended to obtain coverage results for femtocell networks when a typical mobile connects to its nearest BS~\cite{Muk11}.

\subsection{Contributions and Outcomes}
The main contributions of this paper are as follows:

\noindent {\em General $K$-tier downlink model:} In Section II, we define a tractable model for downlink multi-tier networks that captures many (but not all) of the most important  network parameters. The model consists of $K$ independent tiers of PPP distributed BSs, where each tier may differ in terms of the average transmit power, the supported data rate, and the BS density (the average number of BSs per unit area). The plausibility of the model versus planned tiers is verified through comparisons in Section V with an actual $4$G macro-cell (1 tier) network with randomly placed lower tiers.

\noindent{\em $\sinr$ distribution, coverage probability ($\pc$), outage probability ($1-\pc$):} Assuming (i) a mobile user connects to the strongest candidate BS, (ii) that the resulting Signal-to-Interference-plus-Noise-Ratio ($\sinr$) is greater than 1 when in coverage, and (iii) Rayleigh fading, we derive an expression for the probability of coverage (equivalently outage) over the entire network under both open and closed access, which allows a remarkably simple closed-form in the high $\sinr$ regime (where interference power dominates noise power) and is shown to be accurate down to -4dB even under weaker assumptions. When all tiers have the same target $\sinr$ threshold, the coverage probability is the complementary cumulative distribution function (CCDF) of effective received $\sinr$ for an arbitrary randomly located mobile user.

\noindent{\em Average Data Rate:} We derive the average data rate experienced by a randomly chosen mobile when it is in coverage, assuming interference is treated as noise but otherwise that the Shannon bound is achieved, i.e. the average rate in coverage is $\E[\log(1+\sinr)| {\rm coverage}]$. This expression is readily computable but involves an integral so is not closed-form.

Some interesting observations can be made from these results.  For example, we show that when the $\sinr$ targets are the same for all tiers in a dense network (thermal noise power negligible compared to interference power), the coverage (and hence outage) probability does not depend upon the number of tiers or the density of BSs in open access, but that $\pc$ generally decreases with both in closed-access. This means that the trend towards increased density and heterogeneity and the resulting increase in interference need not reduce the typical $\sinr$, as is commonly feared.  On the contrary, aggregate network throughput will increase linearly with the number of BSs since the $\sinr$ statistics will stay the same per cell.

We also provide the average load per tier, which is the average fraction of users served by the BSs belonging to a particular tier or equivalently the probability that a mobile user is served by that tier. In line with intuition, the per-tier load is directly proportional to the density of its BSs and their average transmit power, and inversely proportional to its $\sinr$ target.
\section{System Model}

\subsection{Heterogeneous Cellular Network Model}
We model a HCN as a $K$-tier cellular network where each tier models the BSs of a particular class, such as those of femtocells or pico-cells. The BSs across tiers may differ in terms of the transmit power, the supported data rate and their spatial density.  We assume that the  BSs in the $i$-th tier are spatially distributed as a PPP $\Phi_i$ of density $\lambda_i$, transmit at power $P_i$, and have a $\sinr$ target of $\T_i$. More precisely a mobile can reliably communicate with a BS $x$ in the $i$-th tier only if its downlink $\sinr$ with respect to that BS is greater than $\T_i$.  Thus, each tier can be uniquely defined by the tuple $\{P_i, \T_i, \lambda_i\}$.

The mobiles are also modeled by an independent PPP $\Phi_m$ of density $\lambda_m$\footnote{Since the $\sinr$ distribution is derived as a spatial average, the distribution of users is implicitly assumed to be homogeneous.}. Without loss of generality, we conduct analysis on a typical mobile user located at the origin. The fading (power) between a BS located at point $x$ and the typical mobile is denoted by $h_x$ and is assumed to be i.i.d  exponential (Rayleigh fading).  More complex channel distributions can be considered in this framework, e.g. in \cite{AndBacGan10} a general interference fading model capable of handling any statistical distribution was used, and using the Fourier integral techniques in \cite{BacBla09} general fading to the selected BS can also considered. Both of these generalizations appear to entail significantly decreased tractability, and are left to future work. The standard path loss function is given by $l(x)=\|x\|^{-\alpha}$, where  $\alpha >2$ is the path loss exponent. Hence, the received power at a typical mobile user from a BS located at point $x_i$ (belonging to $i^{th}$ tier) is  $P_i h_{x_i} \|x_i\|^{-\alpha}$, where $h_{x_i} \sim \exp(1)$. The resulting $\sinr$ expression assuming the user connects to this BS is:
\begin{equation}
\sinr(x_i) = \frac{P_i h_{x_i} \|x_i\|^{-\alpha}}{\sum_{j=1}^{K} \sum_{x \in \Phi_j \setminus x_i} P_j h_{x} \|x\|^{-\alpha}+\sigma^2},
\end{equation}
where $\sigma^2$ is the constant additive noise power. One of the ways to set the value of $\sigma^2$ is according to the desired received $\snr$ at the cell-edge. We will comment more on this in  Section \ref{sec:numerical}, where we show that self-interference dominates noise in the typical HCNs.
We assume each mobile user connects to its strongest BS instantaneously, i.e., the BS that offers the highest received $\sinr$. Mathematically the typical node at the origin is in coverage if:
\[  \max_{x \in \Phi_i}\sinr(x) > \T_i, \]
for some $1\leq i\leq K$. An assumption that greatly simplifies the analysis is that $\T_i > 1$ ( 0 dB)\footnote{This assumption is relaxed in~\cite{DhiGanC2011a} where we compute the coverage probability for general $\T_i$. This also enabled us to compute the ergodic rate for a typical mobile user.}. The following Lemma shows that under this assumption, at most one BS in the entire network can provide $\sinr$ greater than the required threshold.  Although some users in commercial cellular networks indeed have operating $\sinr$ below 0 dB, they are in a distinct minority (cell edge users) and in Section V we show numerically that this model holds very accurately at least to $-4$ dB, which covers cell edge users as well. The following Lemma characterizes the number of potential BSs that a mobile can connect to and will be used in the later sections. 
\begin{lemma}
\label{Thm:0}
 Given positive real numbers $\{a_1, a_2 \ldots a_n\}$, which correspond to the received power from each BS at the typical mobile user and defining $c_i = \frac{a_i}{\sum_{j\ne i}a_j +\sigma^2}$, which corresponds to the $\sinr$ of the $i^{th}$ BS, at most $m$ $c_i$'s can be greater than $1/m$ for any positive integer $m$.
\end{lemma}
\begin{proof}
See Appendix.
\end{proof}

\begin{figure}
\centering
\includegraphics[width=0.7\columnwidth]{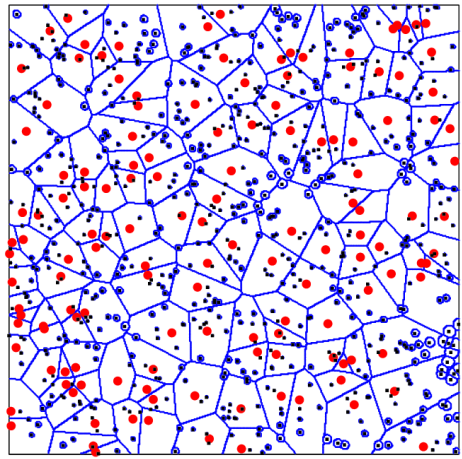}
\caption{Coverage regions in a two-tier network as per the model used in this paper. Both macro (large circles) and femto (small dark squares) BSs are distributed as independent PPPs with $P_1 = 1000P_2$ and $\lambda_2 = 5\lambda_1$.}
\label{fig:2HCN_PPP}
\end{figure}
\begin{figure}
\centering
\includegraphics[width=0.7\columnwidth]{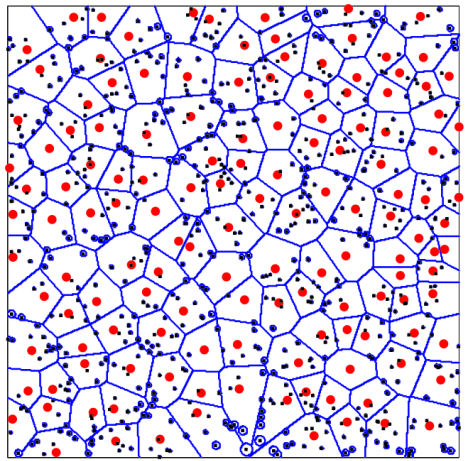}
\caption{Coverage regions in a two-tier network where Macro (tier-1) BS locations (large circles) correspond to actual $4$G deployment. Femto BSs (small dark squares) are distributed as a PPP ($P_1 = 1000P_2$ and $\lambda_2 = 5\lambda_1$).}
\label{fig:2HCN_RealData}
\end{figure}

\subsection{Coverage Regions}
Before going into the analysis and main results, it may be helpful to first build some intuition about the proposed model, and its resulting coverage regions. The illustrative HCN coverage regions can be visually plotted in two steps, resulting in Figs.~\ref{fig:2HCN_PPP}-\ref{fig:3HCN_RealData}. First, we randomly place $K$ different types of BSs on a $2$-D plane according to the aforementioned independent PPPs. Ignoring fading, the space is then fully tessellated following the maximum $\sinr$ connectivity model, which is equivalent to maximum $\sir$ and maximum power connectivity models in the absence of fading. Please note that in reality the cell boundaries are not as well defined as shown in these coverage regions due to fading. Therefore, these plots can be perceived as the average coverage footprints over a period of time so that the effect of fading is averaged out. Due to the differences in the transmit powers over the tiers, these average coverage plots do \emph{not} correspond to a standard Voronoi tessellation (also called a Dirichlet tessellation)~\cite{Aur91}. Instead, they closely resemble a \textit{circular Dirichlet tessellation}, also called a multiplicatively weighted Voronoi diagram \cite{AshBol86}. The coverage regions for a two-tier network -- for example comprising macro and femtocells -- are depicted in Figs.~\ref{fig:2HCN_PPP} and \ref{fig:2HCN_RealData} for two cases: 1) the macro-cell BSs are distributed according to PPP (our model), and 2) the macro-cell BSs correspond to an actual $4$G deployment over a relatively flat urban region. The femtocells are distributed according to an independent PPP in both cases. Qualitatively, the coverage regions are quite similar in the two cases.

In Figs.~\ref{fig:3HCN_PPP} and \ref{fig:3HCN_RealData}, the coverage regions are now shown with an additional pico-cell tier.  As is the case in the actual networks, we assume that the macro-cells have the highest and the femtocells have the lowest transmit power, with pico-cells somewhere in between. For example, in LTE \cite{LTEBook}, typical values are on the order of 50W, .2W, and 2W, respectively.  Therefore, femtocell coverage regions are usually much smaller than the other two tiers, particularly when they are nearby a higher power BS.   Similarly, we observe that the coverage footprint of pico-cells increases when they are farther from the macro BSs. These observations highlight the particularly important role of smaller cells where macrocell coverage is poor.

\begin{figure}
\centering
\includegraphics[width=0.7\columnwidth]{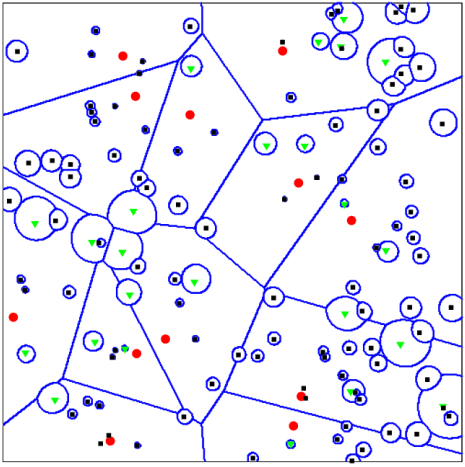}
\caption{Close-up view of coverage regions in a three-tier network. All the tiers, i.e., tier-1 macro (large circles), tier-2 pico (light triangles) and tier-3 femto (small dark squares), are modeled as independent PPPs. $P_1 = 100P_2 = 1000P_3$, $\lambda_3 = 4\lambda_2=8\lambda_1$. }
\label{fig:3HCN_PPP}
\end{figure}

\begin{figure}
\centering
\includegraphics[width=0.7\columnwidth]{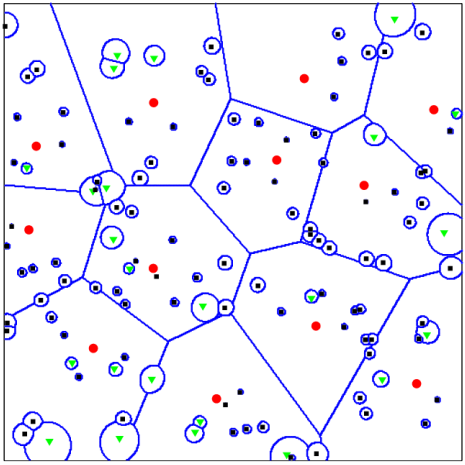}
\caption{Coverage regions in a three-tier network where macro BS locations (large circles) now correspond to actual $4$G deployment. Other parameters are same as Fig.~\ref{fig:3HCN_PPP}}.
\label{fig:3HCN_RealData}
\end{figure}

\subsection{Applicability of the Model}

The model is applicable both to non-orthogonal (CDMA\footnote{For CDMA networks, although the \emph{received} $\sinr$ is generally much smaller than 1, the post-despreading $\sinr$, which is what determines coverage/outage, is often greater than or at least close to 1, so our model and Lemma 1 are still reasonable if the interference term is divided by a spreading factor $M$.}) and orthogonal (TDMA, OFDMA) cellular networks. The analysis is for a single frequency band and assumes that all BSs are transmitting continuously in all time slots at constant power, although if some fraction $f$ of time slots were not used (at random), then the resulting density of interfering BSs would simply be $(1-f)\lambda$ and the analysis could be extended. In OFDMA-based networks, it is desirable to move strongly interfering neighbors or tiers to orthogonal resources in time and/or frequency and so the coverage can be improved.  Similarly, additional enhancements like opportunistic scheduling or multiple antenna communication should increase coverage and/or rate and this framework could be extended to indicate the gains of different approaches. Although we do not explicitly consider antenna sectoring, it can be easily incorporated in the current model if sectoring is done randomly. If the beam is partitioned into $n$ equal sectors, the density of interfering BSs reduces by a factor of $n$ because the probability that the beam of any BS would point towards a randomly chosen BS is $1/n$. Cellular engineers will note that further details are missing from this model. In addition to shadowing, we do not consider frequency reuse, power control, or any other form of interference management, leaving these to future extensions. In short, this is a baseline tractable model for HCNs.

\section{Coverage Probability and Average Load per Tier}

A typical mobile user is said to be in coverage if it is able to connect to at least one BS with $\sinr$ above its threshold. In the case when all the tiers have same $\sinr$ threshold $\T>1$, coverage probability is precisely the complementary cumulative distribution function (CCDF) of the effective received $\sinr$, outage being the CDF, i.e., $1-\text{CCDF}$.  With this understanding, we now derive the probability of coverage for a randomly located mobile user both for \textit{open} and \textit{closed access} networks (defined below). Using these results, we also derive a measure of average load per tier in terms of the fraction of users served by each tier.

\subsection{Open Access}
We first assume the open access strategy where a typical mobile user is allowed to connect to any tier without any restriction. Under the current system model, this strategy reduces to choosing the strongest BS, i.e., the one that delivers the maximum received $\sinr$.

\subsubsection{Coverage Probability}
The main result for the probability of coverage in open access networks is given by Theorem~\ref{Thm:1}.

\begin{theorem}[General case] When $\T_i>1$, the coverage probability for a typical randomly located mobile user in open access is
\begin{align}
&\pc(\{\lambda_i\}, \{\T_i\}, \{P_i\})= \sum_{i=1}^{K} \lambda_i  \int_{\R^2}\exp\Big(-C(\alpha)\left(\frac{\T_i}{P_i}\right)^{2/\alpha}\nonumber\\
&  \|x_i\|^2\sum_{m=1 }^K \lambda_m P_m^{2/\alpha}  \Big)\exp\Big(-\frac{\T_i \N}{P_i }\|x_i\|^\alpha \Big) \nrmd  x_i,
\label{eqn:Pc}
\end{align}
where $C(\alpha) = 2\pi^2\csc(\frac{2\pi}{\alpha}) \alpha^{-1}$.
%
\label{Thm:1}
\end{theorem}

\begin{IEEEproof}
See Appendix.
\end{IEEEproof}
Theorem~\ref{Thm:1} gives a simple and fairly general expression for coverage probability. For better understanding of the proof, we now provide a brief description of the main steps. First recall that a mobile user is in coverage if it is able to connect to at least one BS with $\sinr$ above its threshold. Now assuming, $\T_i>1$, $\forall\ i$, we know from Lemma~\ref{Thm:0} that a mobile can connect to at most one BS. Therefore, $\pc$ can now be defined as the sum of the probabilities that each BS connects to the mobile (with the understanding that all the events are mutually exclusive and at most one of them happens at any time). This leads to a sum of probabilities over PPP, which can be converted to a simple integral of Laplace transform of cumulative interference using Campbell-Mecke Theorem~\cite{StoKen96}. A closed form expression for the Laplace transform can be evaluated in two main steps. Firstly, the nature of interference function (sum over PPP) leads to a product form for its Laplace transform. Using probability generating functional (PGFL) of PPP~\cite{StoKen96} and the fact that fading power is exponentially distributed, we arrive at the closed form expression for Laplace transform which directly leads to the final result of the Theorem. This result can be simplified further for the interference-limited case, where it reduces to a remarkably simple closed-form expression given by Corollary~\ref{Cor:1}.

\begin{cor}[No-noise]
\label{Cor:1}
In an interference limited network,  \ie, when self-interference dominates thermal noise, the coverage probability of a typical mobile user simplifies to
\[\pc(\{\lambda_i\}, \{\T_i\}, \{P_i\}) = \frac{\pi}{C(\alpha)}  \frac{\sum_{i=1}^K \lambda_i P_i^{2/\alpha} \T_i^{-2/\alpha}}{\sum_{i=1}^K \lambda_i P_i^{2/\alpha}   }, \ \ \beta_i>1.\]
\end{cor}
\begin{proof}
Follows from Theorem \ref{Thm:1} with $\sigma^2=0$.
\end{proof}
The simplicity of this result leads to some important observations. Firstly, setting $K=1$ leads to the single-tier case, where the coverage probability is given by:
\begin{equation}
\pc(\lambda, \T, P) = \frac{\pi}{C(\alpha)\T^{2/\alpha}}.
\label{eq:1tierPc}
\end{equation}
From \eqref{eq:1tierPc}, we note that the $\pc$ in an interference-limited single-tier network is independent of the density of the BSs, and is solely dependent upon the target Signal-to-Interference-power-Ratio ($\sir$)\footnote{When the system is interference-limited, $\sinr$ and $\sir$ can be used interchangeably since thermal noise is negligible compared to the interference power. However, for concreteness we will henceforth use $\sir$ when assuming interference-limited network.}. This is consistent with~\cite{AndBacGan10}, where a similar observation was made for a single-tier network using \textit{nearest neighbor connectivity model}. The intuition behind this observation is that the change in the density of BSs leads to the change in the received and interference powers with the same factor and hence the effects cancel.

From Corollary \ref{Cor:1}, it also follows that, if $\T_i = \T$, $\forall\ i$, in an interference-limited network then  $\pc(\{\lambda_i\}, \T, \{P_i\}) = \frac{\pi}{C(\alpha)\T^{2/\alpha}}$.
This is perhaps an unexpected result since it states that the coverage probability is not affected by the number of tiers or their relative densities and transmit powers in an interference-limited network. In fact, it is exactly the same as that of the single-tier case. Therefore, more BSs can be added in any tier without affecting the coverage and hence the net \textit{network capacity} can be increased linearly with the number of BSs. The intuition behind this result is that the decision of a mobile user to connect to a BS depends solely on the received $\sir$ from that BS and a common target $\sir$, unlike the general case where it also depends upon the tier to which the BS belongs. Thus, the mobile user does not differentiate between the tiers when the $\sir$ thresholds are the same. Surprisingly, this leads to a situation similar to the one discussed for the single-tier case above, where the change in the received power due to the change in the density or the transmit power of BSs of some tier is equalized by the change in the interference power. Significantly, this implies that the interference from smaller cells, such as femtos and picos, need not decrease network performance in open access networks.

\subsubsection{Average Load per Tier}
The average load on each tier is defined as the average fraction of users in coverage served by that tier. This can also be interpreted as the average fraction of time for which each mobile is connected to the BSs belonging to a particular tier. The main result for the average load per tier in \textit{open access} is given by Proposition~\ref{Thm:3}.
\begin{prop}
\label{Thm:3}
The average fraction of users served by $j^{th}$ tier (also the average load on $j^{th}$ tier) in open access is
\begin{align*}
\bar{N}_j &= \frac{\lambda_j}{\pc(\{\lambda_i\}, \{\T_i\}, \{P_j\})}     \int_{\R^2}\exp\Big(-C(\alpha)\left(\frac{\T_i}{P_i}\right)^{2/\alpha}\nonumber\\
&||x_i||^2\sum_{m=1 }^K \lambda_m P_m^{2/\alpha}  \Big)\exp\Big(-\frac{\T_i \N}{P_i }||x_i||^\alpha \Big) \nrmd  x_i.
\end{align*}
%
\end{prop}
\begin{IEEEproof}
See Appendix.
\end{IEEEproof}
In an interference-limited scenario, this result reduces to a simple closed form expression, which is given by the following Corollary.
\begin{cor}
\label{Cor:2}

When noise is neglected, \ie, $\sigma^2 =0$,
\[\bar{N}_j =\frac{\lambda_j P_j^{2/\alpha} \T_j^{-2/\alpha}}{\sum_{i=1}^K \lambda_i P_i^{2/\alpha} \T_i^{-2/\alpha}}.\]
\end{cor}
From Corollary~\ref{Cor:2}, we observe that the load on each tier is directly proportional to the quantity $\lambda_j P_j^{2/\alpha} \T_j^{-2/\alpha}$. In line with intuition, a tier will serve more users if it has a higher BS density or higher transmit power or a lower $\sir$ threshold. When the thresholds of all tiers are equal to $\T$ and the transmit powers of all BSs equal to $P$, the average load on each tier is
$\bar{N}_j = \frac{\lambda_j}{\sum_{i=1}^K\lambda_i}$.
Hence, as expected the average load on each tier is directly proportional to the density of its BSs.

\subsection{Closed Access}
Under closed access, also known as a closed subscriber group, a mobile user is allowed to connect to only a subset of tiers and the rest of the tiers act purely as interferers.  The motivation for closed access particularly applies to privately owned infrastructure, such as femtocells or perhaps custom picocells mounted on a company's roof to improve service to their staff.  The desirable aspects of closed access can include protection of finite backhaul capacity, security, and the reduction in the frequency of handoffs experienced by mobile users and the associated overhead required.  In the context of our model, closed access means that if the strongest BS lies in the restricted tier, it by definition leads to an outage event irrespective of the received $\sinr$ associated with that BS.  Furthermore, since closed access is a constraint on connectivity, it would naturally lead to reduced coverage probability.  This intuition is verified in the following discussion.

\subsubsection{Coverage Probability}
The main result of coverage probability in closed access networks is given by Lemma \ref{Thm:5}.
\begin{lemma}
\label{Thm:5}
When a typical  mobile user is allowed to connect to only a subset $\mathcal{B} \subset \{1,2,,\hdots,K\}$,  the coverage probability for closed access is
\begin{align}
&\pc(\{\lambda_i\}, \{\T_i\}, \{P_i\})= \sum_{i\in \mathcal{B}} \lambda_i  \int_{\R^2}\exp\Big(-C(\alpha)\left(\frac{\T_i}{P_i}\right)^{2/\alpha}\nonumber\\
&  \|x_i\|^2\sum_{m=1 }^K \lambda_m P_m^{2/\alpha}  \Big)\exp\Big(-\frac{\T_i \N}{P_i }\|x_i\|^\alpha \Big) \nrmd  x_i.
%
\label{eqn:Pc_closed}
\end{align}
\end{lemma}
\begin{proof}
The coverage probability is
 \begin{align}
\pc(\{\lambda_i\}, \{\T_i\}, \{P_i\}) &= \P \left(\bigcup_{i\in \mathcal{B}, x_i \in \Phi_i} \sinr(x_i) > \T_i \right) \nonumber\\
& \stackrel{(a)}{=}  \sum_{i\in \mathcal{B} }\E \sum_{x_i \in \Phi_i}  \left[\nb1 \left( \sinr(x_i) > \T_i  \right)  \right], \nonumber
\end{align}
where $(a)$ again follows from Lemma~\ref{Thm:0} under the assumption that $\T_i>1$. Following the same steps as the proof of Theorem \ref{Thm:1}, we arrive at the final result.
\end{proof}

The following corollary specializes from Lemma \ref{Thm:5} to interference-limited HCNs.
\begin{cor}
When $\N=0$,
\[\pc(\{\lambda_i\}, \{\T_i\}, \{P_i\})=\frac{\pi}{C(\alpha)}  \frac{ \sum_{i\in \mathcal{B}}\lambda_i P_i^{2/\alpha} \T_i^{-2/\alpha}}{   \sum_{i=1}^K \lambda_i P_i^{2/\alpha}   }.\]
\end{cor}
If the threshold of each tier to be same (and equal to $\T$) and the transmit power of each tier to be same (and equal to $P$), the coverage probability is
$\frac{\pi}{C(\alpha) \T^{2/\alpha}  }  \frac{\sum_{i\in \mathcal{B}} \lambda_i}{  \sum_{i=1}^K \lambda_i  }$. So, if the thresholds and transmit powers of all the tiers are same, closed access has a lower coverage than open access by a factor of $\frac{\sum_{i\in \mathcal{B}} \lambda_i}{  \sum_{i=1}^K \lambda_i  }$.

\subsubsection{Average Load per Tier}
The main result for the average load per tier under closed access is given by Proposition \ref{Prop:1}. The proof directly follows from the proof of Proposition~\ref{Thm:3} with the understanding that the coverage event would now be defined by only the ``allowed'' tiers.
\begin{prop}
\label{Prop:1}
The average fraction of users in coverage served by $j^{th}$ tier (also the average load on $j^{th}$ tier) in closed access is
\begin{equation}
\bar{N}_j = \left\{ \begin{array}{ll}
\frac{\lambda_j\delta_j}{\pc(\{\lambda_i\}, \{\T_i\}, \{P_j\})} & j\in \mathcal{B},\\
0 & \text{otherwise}.
\end{array}
\right.
\end{equation}
where $\pc(\{\lambda_i\}, \{\T_i\}, \{P_j\})$ is the coverage probability under closed access given by Lemma~\ref{Thm:5} and 
\[\delta_j=\int_{\R^2}e^{-\left(\frac{\T_j}{P_j}\right)^{2/\alpha}C(\alpha)||x||^2 \sum_{m=1 }^K \lambda_m P_m^{2/\alpha} }e^{-\frac{\T_j \N}{P_j }||x||^\alpha } \nrmd  x.\]
The corresponding result for the interference-limited networks is
\begin{equation}
\bar{N}_j = \left\{ \begin{array}{ll}
\frac{\lambda_j P_j^{2/\alpha} \T_j^{-2/\alpha}}{\sum_{i\in \mathcal{B}} \lambda_i P_i^{2/\alpha} \T_i^{-2/\alpha}} & j \in \mathcal{B},\\
0 & \text{otherwise}.
\end{array}
\right.
\end{equation}
\end{prop}

\section{Average Rate}
In this section, we derive the average rate  $\bar{R}$ achievable by a random mobile user when it is in coverage both for the open and closed access strategies. It is worth noting that since the rate is computed conditioned on the mobile being in coverage, it is not the same as the classic ergodic rate $\E[R]$. The motivation behind considering this metric is that given the coverage/outage information, the service providers are interested in knowing the average rate they can provide to the users that are in coverage.
\subsection{Open Access}
The main result for the average rate in open access is given in Theorem \ref{Thm:4}. In this section, for notational simplicity, we restrict our attention to the case of $\N=0$. However, the results can be extended to the general case with noise in a straightforward manner.
\begin{theorem}
The average rate achievable by a randomly chosen mobile in open access when it is in coverage is
\begin{equation}
\bar{R} = \log\left( 1+ \T_{min}\right) + \frac{\sum_{i=1}^K \lambda_i P_i^{2/\alpha} \calA(\alpha, \T_i, \T_{min})  }{\sum_{i=1}^K \lambda_i P_i^{2/\alpha} \T_i^{-2/\alpha}},
\end{equation}
where
\[\calA(\alpha, \T_i, \T_{min}) = \int_{\T_{min}}^{\infty} \frac{\max(\T_i, x)^{-2/\alpha}}{1+x} \nrmd  x,\]
and $\T_{min} = \min\{\T_1, \T_2, \ldots, \T_K\}$.
\label{Thm:4}
\end{theorem}

\begin{IEEEproof}
See Appendix.
\end{IEEEproof}
Thus we observe that the average rate expression involves only a single integral which can be easily evaluated numerically.

\begin{cor}
\label{Cor:4}
Using the same threshold $\T$ for all tiers, the average rate achievable by a randomly chosen mobile that is in coverage in open access is:
\begin{equation}
\bar{R} = \log(1+\T) + \T^{2/\alpha} \calA(\alpha, \T, \T).
\end{equation}
\end{cor}
The above result shows that the average rate is independent of the density of BSs of each tier when the $\sir$ thresholds are same for all the tiers. This is expected because the distribution of $\max \sir$ does not depend upon the density of BSs in this case (follows from Theorem \ref{Thm:1}).

\subsection{Closed Access}
The average rate $\bar{R}_c$ achievable by a randomly chosen mobile under closed access (assuming it is under coverage) can be expressed as:
\begin{equation}
\E \left[ \log \left( 1+ \max_{x \in \bigcup_{i \in \mathcal{B}} \Phi_i } (\sir(x))\right) \Big| \bigcup_{i \in \mathcal{B} } \bigcup_{x \in \Phi_i} \left( \sir(x) > \T_i  \right)   \right].
\end{equation}
Following the same steps as in proof of Theorem \ref{Thm:4}, we arrive at the following Proposition.

\begin{prop}
Assuming a mobile user is allowed to connect to only a subset $\mathcal{B}$ of the $K$ tiers, the average rate (assuming mobile is under coverage) can be expressed as:
\begin{equation}
\bar{R}_c = \log\left( 1+ \T_{min}\right) + \frac{\sum_{i\in \mathcal{B}} \lambda_i P_i^{2/\alpha} \calA(\alpha, \T_i, \T_{min})  }{\sum_{i\in \mathcal{B}} \lambda_i P_i^{2/\alpha} \T_i^{-2/\alpha}},
\end{equation}
where $\T_{min} = \min\limits_{i \in \mathcal{B}}\{\T_i\}$.
\end{prop}

\begin{cor}
\label{Cor:6}
Assuming the threshold of each tier is the same and equal to $\T$, the average rate achievable by a randomly chosen mobile in coverage under closed access is
\begin{equation}
\bar{R}_c = \log(1+\T) + \T^{2/\alpha} \calA(\alpha, \T, \T).
\end{equation}
\end{cor}
From Corollaries \ref{Cor:4} and \ref{Cor:6}, we observe that the average rate ($\bar{R}$) of the mobile while it is in coverage is not affected by access control when the thresholds are the same for all tiers. However, since the coverage probability is lower in case of closed access, it would naturally lead to a lower ergodic rate as compared to the open access networks. Interested readers can refer to~\cite{DhiGanC2011a} for the derivation of ergodic rate in this framework.

\section{Numerical Results}
\label{sec:numerical}
\begin{figure}[ht!]
\centering
\includegraphics[width=0.9\columnwidth]{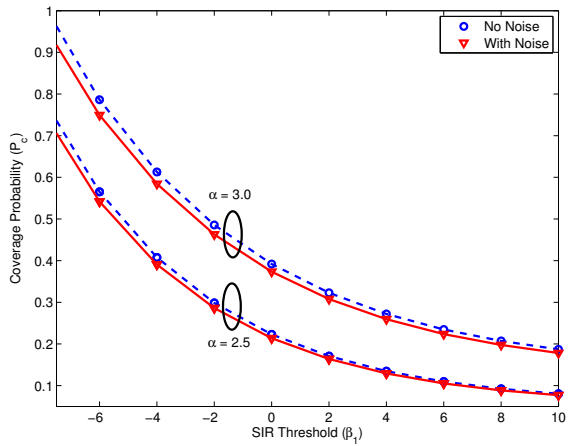}
\caption{Coverage probability in a two-tier HCN with and without thermal noise ($K=2$, $P_1 = 25 P_2$, $\lambda_2 = 5\lambda_1$, $\T_2 = 1$ dB, $\snr_{edge} = 0$ dB).}
\label{fig:PcK2NoiseComp}
\end{figure}

\begin{figure}[ht!]
\centering
\includegraphics[width=0.9\columnwidth]{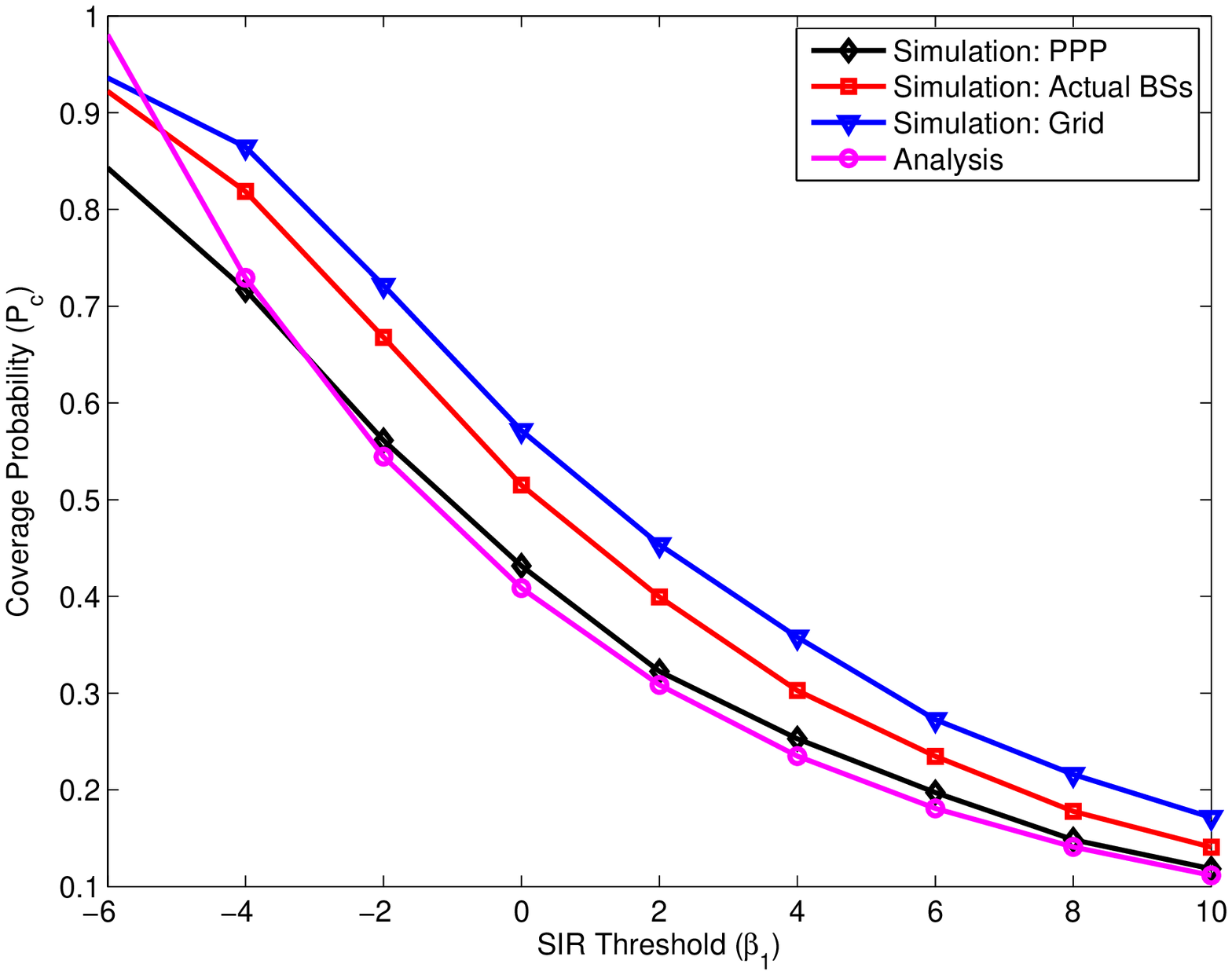}
\caption{Coverage probability in a two-tier HCN ($K=2$, $\alpha=3$, $P_1 = 100 P_2$, $\lambda_2 = 2\lambda_1$, $\T_2 = 1$ dB, No noise).}
\label{fig:PcK2NoNoiseA3}
\end{figure}

\begin{figure}[ht!]
\centering
\includegraphics[width=0.9\columnwidth]{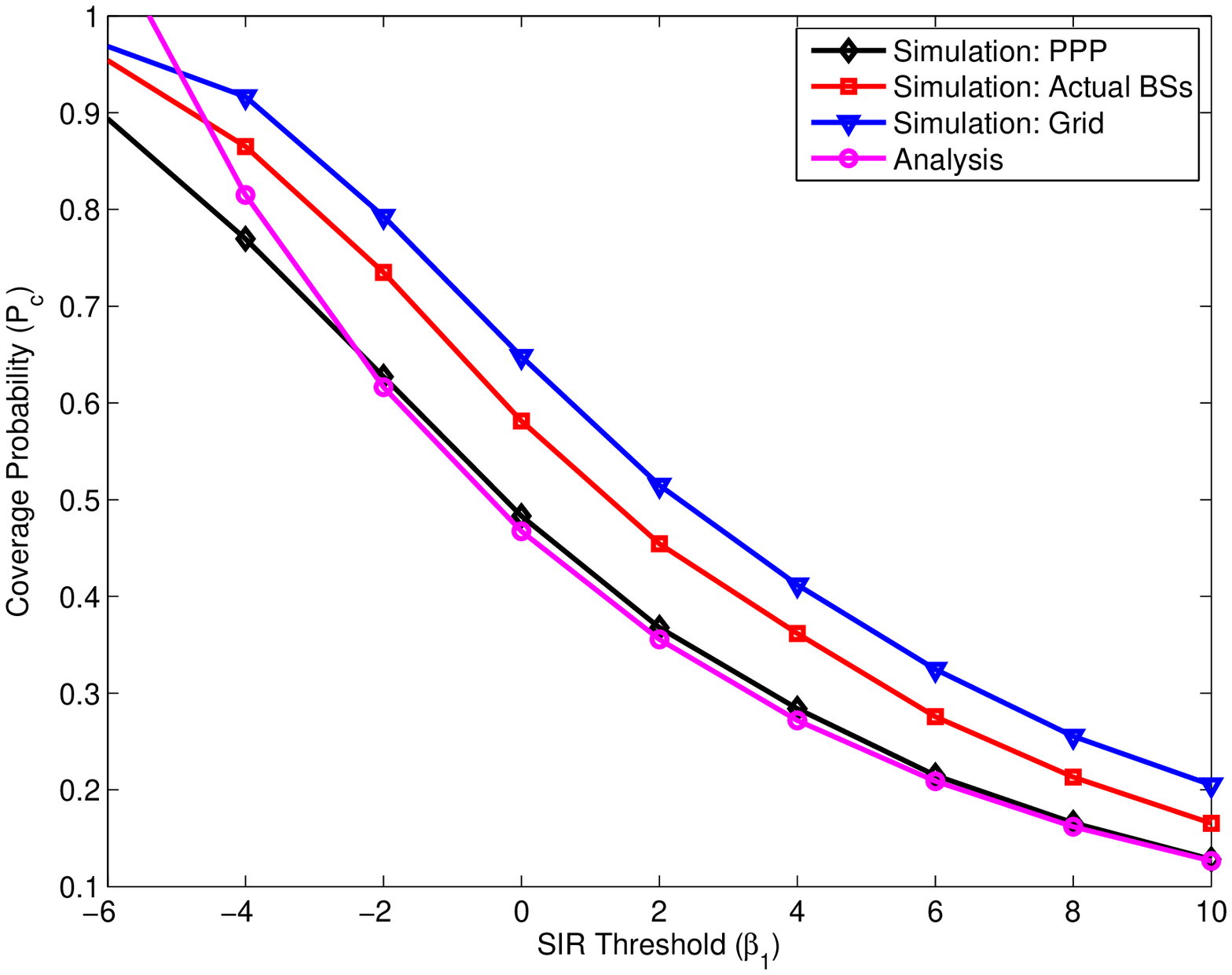}
\caption{Coverage probability in a two-tier HCN ($K=2$, $\alpha=3.2$, $P_1 = 1000 P_2$, $\lambda_2 = 4\lambda_1$, $\T_2 = 1$ dB, No noise).}
\label{fig:PcK2NoNoiseA3_2}
\end{figure}
Most of the analytical results presented in this paper are fairly self-explanatory and do not require a separate numerical interpretation. Therefore, to avoid repetition, we will present only non-obvious trends and validation of the model in this section.
\subsection{Effect of Thermal Noise}
We first study the effect of thermal noise on the coverage probability by considering a typical two-tier network consisting of macro-cells overlaid with pico-cells. To set the noise power, we use the following notion of cell-edge users in this example. Defining the distance of the the nearest macro BS to the typical mobile user to be $d$ and the underlying random variable to be $D$, the mobile user is said to be on the cell edge if $\P(D\le d) \ge P_{edge}$, where $P_{edge}$ is set to $0.9$ for this illustration. For PPP($\lambda$), $\P(D\le d)= 1 - \exp(\lambda \pi d^2)$, giving $d \ge \sqrt \frac{-\ln(1-P_{edge})}{\pi \lambda}$. For a desired edge-user $\snr$, say $\snr_{edge}$, $\sigma^2$ can be approximated as $\sigma^2 \approx  \frac{P_t d_{edge}^{-\alpha}}{\snr_{edge}}$, where $d_{edge}$ is the limiting value of $d$ evaluated above. Under this setup, we present the coverage probability for various values of $\alpha$ in Fig.~\ref{fig:PcK2NoiseComp}. By comparing these results with the no-noise case, we note that the typical HCNs are interference limited and hence thermal noise has a very limited effect on coverage probability. Therefore, we will ignore noise in the rest of this section.

\subsection{Validity of PPP Model and $\T>1$ Assumption}
While a random PPP model is probably the best that can be hoped for in modeling ``unplanned'' tiers, such as femtocells, its  accuracy in modeling ``planned'' BS locations, such as those of macro-cells, is open to question. Therefore we verify the PPP assumption for macro-cells from a coverage probability perspective by considering a two-tier network in three different scenarios: 1) the macro-cell BSs are distributed according to PPP (our model), 2)  the macro-cell BSs correspond to an actual $4$G deployment, and 3) macro-cell BSs are distributed according to hexagonal grid model. The second tier is modeled as an independent PPP in all three cases. As shown in Figs.~\ref{fig:PcK2NoNoiseA3} and~\ref{fig:PcK2NoNoiseA3_2}, the actual coverage probability lies between the coverage probabilities of the PPP and grid model. This is because the likelihood of a dominant interferer is highest for the PPP and lowest for the grid model. This comparison shows that the PPP assumption is nearly as accurate as the grid model in the case of macro-cells, with the PPP providing a lower bound and grid model providing an upper bound to the actual coverage probability.

We now focus on the $\T>1$ assumption by comparing the theoretical and simulated results for coverage probability in Figs.~\ref{fig:PcK2NoNoiseA3} and \ref{fig:PcK2NoNoiseA3_2}. As expected, the simulated and analytical results match reasonably well for $\T_i>1$ but interestingly, the theoretical results also provide a tight upper bound to the exact solution even until about $\T_1 = -4$ dB ($\approx .4$). Therefore, the analytical results also cover typical cell edge users. The same trend is observed in the case of average rate results presented in Fig.~\ref{fig:RateK2NoNoise}, which are also accurate down to about $-4$ dB target-$\sir$.
\begin{figure}[ht!]
\centering
\includegraphics[width=.85\columnwidth]{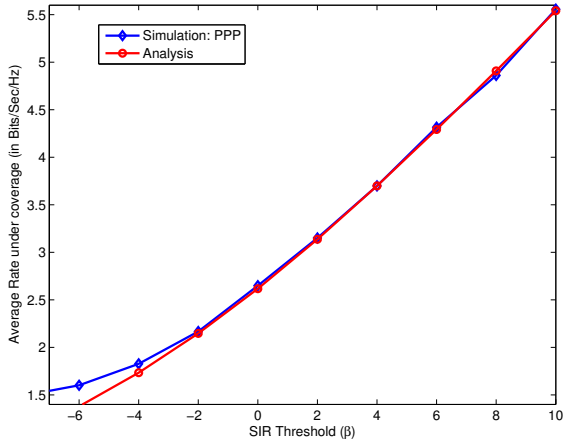}
\caption{Average rate while mobile is in coverage ($K=2$, $\alpha=3$, $P_1 = 1000 P_2$, $\lambda_2 = 2\lambda_1$, $\T_1 = \T_2 = \T$, no noise, open access).}
\label{fig:RateK2NoNoise}
\end{figure}

\section{Conclusions}
We have provided a new tractable model for $K$-tier downlink HCNs and an associated analysis procedure that gives simple mathematical expressions for the most important performance metrics. The possible extensions of this work are numerous, and could include physical layer technologies like multiple antennas, spread spectrum, power control, interference cancelation or interference alignment. For the MAC layer, it would be useful to include scheduling, resource allocation and/or various forms of frequency reuse (such as fractional frequency reuse) in this setup. At the network level, further BS cooperation techniques could use the framework.

The random spatial model used in this paper could likely be further improved by incorporating a point process that models repulsion or minimum separation distance between BSs, such as determinantal and Mat\'ern processes~\cite{StoKen96, BacBlaNOW}, respectively.  Often picocells or enterprise (operator deployed) femtocells are clustered in high-demand areas, so a Poisson cluster process \cite{GanHae09a} might be useful to model such scenarios.  It would also be very helpful to consider further actual deployments from the one considered in this paper (flat, large, low density urban), and see which of these general point processes best model different deployment scenarios such as dense urban (New York, Tokyo) vs. sprawling (Los Angeles, Sydney), flat relatively uniform cities (Paris, London) vs. mountainous/coastal cities (Rio de Janiero, Hong Kong), and so on. Finally, the user distribution was implicitly assumed in this paper to be homogeneous, with all BSs actively transmitting at all times.  An important extension is to consider non-homogeneous user distributions and realistic traffic models to better understand optimal cell association policies in heterogeneous cellular networks.  An initial characterization of biased cell association policies can be found in \cite{JoSanXia11}.

\appendix
\subsection{Proof of Lemma \ref{Thm:0}}
Since $\sinr < \sir$ for all the BSs, defining $b_i = \frac{a_i}{\sum_{j\ne i}a_j}$ as the $\sir$ corresponding to the $i^{th}$ BS, it suffices to show that at most $m$ $b_i$'s can be greater than $1/m$ for any positive integer $m$. This is shown below.
\begin{align}
b_i &= \frac{a_i}{\sum_{j\ne i}a_j} = \frac{a_i}{\sum_{j}a_j - a_i} \nonumber \\
& \Rightarrow  \frac{b_i}{1+b_i} = \frac{a_i}{\sum_{j}a_j} \nonumber \\
& \Rightarrow  \sum_{i=1}^n \frac{1}{1/b_i+1} = 1.
\label{eqn:TProof}
\end{align}
We first prove the result for $m=1$ (by contradiction) and then show that it can be trivially extended to the case of general $m$.
We first observe that \eqref{eqn:TProof} is satisfied if only one of the $b_i$'s is greater than 1. Now assume that two $b_i$'s are greater than one and without loss of generality, assume that they are $b_1$ and $b_2$. This implies $1/b_1$ and $1/b_2 \in (0, 1)$. Therefore, $\frac{1}{1/b_i+1}$ and $\frac{1}{1/b_i+1} \in (1/2, 1)$. Thus,
\begin{align}
\sum_{i=1}^n \frac{1}{1/b_i+1} &= \sum_{i=1}^2 \frac{1}{1/b_i+1} + \sum_{i=3}^n \frac{1}{1/b_i+1}, \nonumber \\
&> 1 + \sum_{i=3}^n \frac{1}{1/b_i+1},
\label{eqn:m1case}
\end{align}
which is in contradiction with \eqref{eqn:TProof}. Since \eqref{eqn:TProof} does not even hold for two $b_i$'s greater than one, it proves that the only one of the $b_i$'s can be greater than one. Similarly for the case of general $m$, it is easy to observe that \eqref{eqn:TProof} is trivially satisfied if at most $m$ of the $b_i$'s are greater than $1/m$. Now assume that $m+1$ $b_i$'s are greater than $1/m$ and without loss of generality, assume that they are $b_1$, $b_2$, \ldots, $b_{m+1}$. Proceeding as in \eqref{eqn:m1case},
\begin{equation}
\sum_{i=1}^n \frac{1}{1/b_i+1} > 1 + \sum_{i=m+2}^n \frac{1}{1/b_i+1},
\end{equation}
which is in contradiction to \eqref{eqn:TProof}. Therefore, at most $m$ $b_i$'s can be greater than $1/m$.
\hfill \IEEEQEDclosed   

\subsection{Proof of Theorem \ref{Thm:1}}
For notational simplicity, denote the set $\{1, 2, \ldots K\}$ by $\calK$. The coverage probability in a $K$-tier network under \textit{maximum $\sinr$ connectivity model} can be derived as follows:
 \begin{align}
&\pc(\{\lambda_i\}, \{\T_i\}, \{P_i\})\nonumber \\
 &= \P \left(\bigcup_{i\in \calK, x_i \in \Phi_i} \sinr(x_i) > \T_i \right) \nonumber\\
&= \E\left[ \nb1 \left( \bigcup_{i\in \calK, x_i \in \Phi_i} \sinr(x_i) > \T_i  \right) \right] \nonumber \\
& \stackrel{(a)}{=}  \sum_{i=1}^{K}\E \sum_{x_i \in \Phi_i}  \left[\nb1 \left( \sinr(x_i) > \T_i  \right)  \right] \nonumber \\
&\stackrel{(b)}{=} \sum_{i=1}^{K} \lambda_i\int_{\R^2}\P\left(\frac{P_i h_{x_i} l(x_i)x_i)}{I_{x_i}+\N} > \T_i  \right) \nrmd  x_i \nonumber  \\
&\stackrel{(c)}{=}
 \sum_{i=1}^{K} \lambda_i \int_{\R^2} \ncalL_{I_{x_i}} \left(\frac{\T_i}{P_i l(x_i)}  \right) e^{\frac{-\T_i \N}{P_i l(x_i)x_i)}} \nrmd  x_i,
\label{eqn:Pc_der}
\end{align}
where $(a)$ follows from Lemma~\ref{Thm:0} under the assumption that $\T_i > 1\ \forall\ i$, $(b)$ follows from Campbell Mecke Theorem~\cite{StoKen96}, and $(c)$ follows from the fact that the channel gains are assumed to be Rayleigh distributed. Here $\ncalL_{I_{x_i}}\left(.\right)$ is the Laplace transform of the cumulative interference from all the tiers when the randomly chosen mobile user is being served by the $i^{th}$ tier. Since the point processes are stationary, the interference does not depend on the location $x_i$. Therefore, we denote $\ncalL_{I_{x_i}}$ by $\ncalL_{I_i}$, which is given by
\begin{align}
\ncalL_{I_i}\left(s\right)&=  \prod_{j=1}^K \E_{I_i} \left[ \prod_{x_j\in\Phi_j / x_i} \exp\left(-s P_j h_{x_j} l(x_j) \right)  \right]. \nonumber 
\end{align}
Using the independence of the fading random variables $\ncalL_{I_i}\left(s\right)$ equals
\begin{align}
&    \prod_{j=1}^K \E_{\Phi_j}\left[\prod_{x_j\in\Phi_j / x_i} \E_{h}\left[\exp \left( -s P_j h_{x_j} l(x_j) \right)  \right]  \right] \nonumber \\
& \stackrel{(a)}{=}  \prod_{j=1}^K \E_{\Phi_j}\left[\prod_{x_j\in\Phi_j / x_i} \frac{1}{1+sP_j l(x_j)}  \right] \nonumber \\
& \stackrel{(b)}{=}  \prod_{j=1}^K \exp  \left( -\lambda_i \int_{\R^2} \left(1-\frac{1}{1+sP_j ||x_j||^{-\alpha}}   \right) \nrmd  x_j  \right)  \nonumber \\
& \stackrel{(c)}{=} \prod_{j=1}^K \exp  \left( -2 \pi \lambda_i (s P_j)^{2/\alpha} \int_{0}^{\infty} r \int_{0}^{\infty} e^{ \left( -t (1+r^{\alpha})\right)}\nrmd  t\ \nrmd  r   \right)
\label{eqn:L_der}
\end{align}
where  $(a)$ follows from the Rayleigh fading assumption (i.e., $h\sim \exp(1)$), $(b)$ follows from probability generating functional (PGFL) of PPP~\cite{StoKen96} and, $(c)$ results from algebraic manipulation after converting from Cartesian to polar coordinates Using some properties of Gamma function, \eqref{eqn:L_der} can be further simplified to
\begin{equation}
\ncalL_{I}(s) = \exp\left(-s^{2/\alpha}C(\alpha)\sum_{i=1 }^K \lambda_i P_i^{2/\alpha}\right),
\label{eqn:L}
\end{equation}
where $C(\alpha) = \frac{2\pi^2\csc(\frac{2\pi}{\alpha})}{\alpha}$.
Using \eqref{eqn:Pc_der} and \eqref{eqn:L} the coverage probability $\pc(\{\lambda_i\}, \{\T_i\}, \{P_i\})$ is
\begin{align*}
\sum_{i=1}^{K} \lambda_i  
& \int_{\R^2}e^{-\left(\frac{\T_i}{P_i}\right)^{2/\alpha}C(\alpha)||x_i||^2\sum_{m=1 }^K \lambda_m P_m^{2/\alpha}}  e^{\frac{-\T_i \N}{P_i }||x_i||^\alpha}\nrmd  x_i,
\end{align*}
which completes the proof.
\hfill \IEEEQEDclosed

\subsection{Proof of Proposition~\ref{Thm:3}}
Let $B_n \in \R^2$ denote an increasing sequence of convex sets with $B_n \subset B_{n+1}$ and $\lim_{n\rightarrow \infty} |B_n| = \infty$. For this proof, we denote $\sir_{x_m}(x_b)$ as the received $\sir$ when the mobile is located at $x_m \neq 0$ connects to BS located at $x_b$. Please recall that the subscript is dropped and $\sir$ is denoted as $\sir(x_b)$ when the mobile user is located at the origin. The average fraction of users served by the $j^{th}$ tier can now be expressed as:
\begin{align}
&\bar{N}_j = \lim_{n\rightarrow \infty} \frac{1}{|B_n|} \sum_{x_m\in B_n \bigcap \Phi_m} \nonumber \\
& \nb1 \left( \bigcup_{x_j \in \Phi_j}\sinr_{x_m}(x_j) > \T_j \Big| \bigcup_{i\in \calK, x_i \in \Phi_i} (\sinr_{x_m}(x_i) > \T_i) \right)   \nonumber \\
& \stackrel{(a)}{=} \P^{!o} \left(\bigcup_{x_j \in \Phi_j} \sinr(x_j) > \T_j \Big|  \bigcup_{i\in \calK, x_i \in \Phi_i} (\sinr(x_i) > \T_i) \right) \nonumber \\
& \stackrel{(b)}{=} \frac{\P \left(\bigcup\limits_{x_j \in \Phi_j} \sinr(x_j) > \T_j,\  \bigcup\limits_{i\in \calK, x_i \in \Phi_i} (\sinr(x_i) > \T_i) \right)}{\P \left(\bigcup_{i\in \mathcal{K}, x_i \in \Phi_i} (\sinr(x_i) > \T_i) \right)}   \nonumber \\
& = \frac{\P \left(\bigcup_{x_j \in \Phi_j} \sinr(x_j) > \T_j\right)}{\P \left(\bigcup_{i\in \calK, x_i \in \Phi_i} (\sinr(x_i) > \T_i) \right)}
\label{eq:avgload}
\end{align}
where $(a)$ follows from the stationarity and the ergodicity of  PPP \cite{StoKen96}. $\P^{!o}$ denotes the reduced Palm distribution of a PPP and  $(b)$ follows from the Slivinak's theorem \cite{StoKen96, KinBook} and  Bayes rule. Noting that $\P \left(\bigcup_{x_j \in \Phi_j} \sinr(x_j) > \T_j\right)$ is the probability of coverage with a single tier $j$, the result follows from Theorem \ref{Thm:1}.
\hfill \IEEEQEDclosed

\subsection{Proof of Theorem \ref{Thm:4}}
Denoting the coverage event $\bigcup_{i=1}^K \bigcup_{x \in \Phi_i} \left( \sir(x) > \T_i  \right)$ by $\mathbf{C}(\{\T_i\})$, the average rate achievable by a randomly chosen mobile user when it is under coverage can be expressed as:
\begin{equation}
\bar{R} = \E \left[ \log \left( 1+ \max_{x \in \bigcup \Phi_i } (\sir(x))\right) \Big|  \mathbf{C}(\{\T_i\})  \right].
\end{equation}
We first derive the conditional complementary cumulative density function (CCDF) of $\max_{x \in \bigcup \Phi_i } (\sir(x))$ as follows:
\begin{align}
&\P \left( \max_{x \in \bigcup \Phi_i } (\sir(x)) > T\  \Big | \ \mathbf{C}(\{\T_i\})   \right)  \nonumber \\
& \stackrel{(a)}{=} \frac{\P \left( \max_{x \in \bigcup \Phi_i } (\sir(x)) > T, \mathbf{C}(\{\T_i\})\right)}{\P(\mathbf{C}(\{\T_i\}))} \nonumber \\
&\stackrel{(b)}{=} \frac{\P \left( \mathbf{C}(\{T\}), \mathbf{C}(\{\T_i\})   \right)}{\P(\mathbf{C}(\{\T_i\}))}, \nonumber \\
&= \frac{\P \left( \mathbf{C}(\{\max(T,\T_i)\})   \right)}{\P(\mathbf{C}(\{\T_i\}))}, \nonumber \\
&\stackrel{(c)}{=}   \left\{ \begin{array}{ccc}
 \frac{\sum_{i=1}^K \lambda_i P_i^{2/\alpha} \max(\T_i, T)^{-2/\alpha}}{\sum_{i=1}^K \lambda_i P_i^{2/\alpha} \T_i^{-2/\alpha}} &; & T > \T_{min}\\
 1 & ; & \text{otherwise}
  \end{array}  \right.,
\label{eq:ccdf}
\end{align}
where $(a)$ follows from Bayes' theorem, $(b)$ follows from Lemma~\ref{Thm:0} under the assumption $\T_i>1\ \forall\ i$, $(c)$ follows from Theorem \ref{Thm:1}, and $\T_{min}$ denotes $\min\{\T_1, \T_2, \ldots, \T_K\}$.

Denoting random variable $\max_{x \in \bigcup \Phi_i } (\sir(x))$ by $X$, $\bar{R}$ can be evaluated as follows:
\begin{eqnarray}
\bar{R} &=& \int_{0}^{\infty} \log(1+x) f_X(x \mid \mathbf{C}(\{\T_i\})) \nrmd  x, \nonumber \\
& = & \int_{x=0}^{\infty} \int_{y=0}^{x} \frac{1}{1+y} f_X(x \mid \mathbf{C}(\{\T_i\}))\ \nrmd  y\ \nrmd  x, \nonumber \\
& \stackrel{(a)}{=} & \int_{y=0}^{\infty} \left( \int_{x=y}^{\infty}  f_X(x \mid \mathbf{C}(\{\T_i\}))\ \nrmd  x \right) \frac{1}{1+y}\ \nrmd  y, \nonumber \\
& = & \int_{0}^{\infty} \frac{\P\left(X>y\mid \mathbf{C}(\{\T_i\}) \right)}{1+y}  \nrmd  y,
\label{eq:Rate1}
\end{eqnarray}
where $(a)$ follows from changing the order of integration. Now we substitute \eqref{eq:ccdf} in \eqref{eq:Rate1} to get the average rate as:
\begin{align}
\bar{R} &= \int_{0}^{\T_{min}} \frac{1}{1+y}d y\ +\nonumber\\
& \frac{1}{\sum_{i=1}^K \lambda_i P_i^{2/\alpha} \T_i^{-2/\alpha}}  \sum_{i=1}^K \lambda_i P_i^{2/\alpha} \int_{\T_{min}}^{\infty} \frac{\max(\T_i, x)^{-2/\alpha}}{1+x} \nrmd  x \nonumber \\
&= \log(1+\T_{min}) + \nonumber \\
& \frac{1}{\sum_{i=1}^K \lambda_i P_i^{2/\alpha} \T_i^{-2/\alpha}}  \sum_{i=1}^K \lambda_i P_i^{2/\alpha} \int_{\T_{min}}^{\infty} \frac{\max(\T_i, x)^{-2/\alpha}}{1+x} \nrmd  x.
\end{align}
This completes the proof.
\hfill \IEEEQEDclosed

\bibliographystyle{IEEEtran}
\bibliography{DhiGanAnd_HCN_JSAC2012}

\begin{biography}[{\includegraphics[width=1in,height
=1.25in,clip,keepaspectratio]{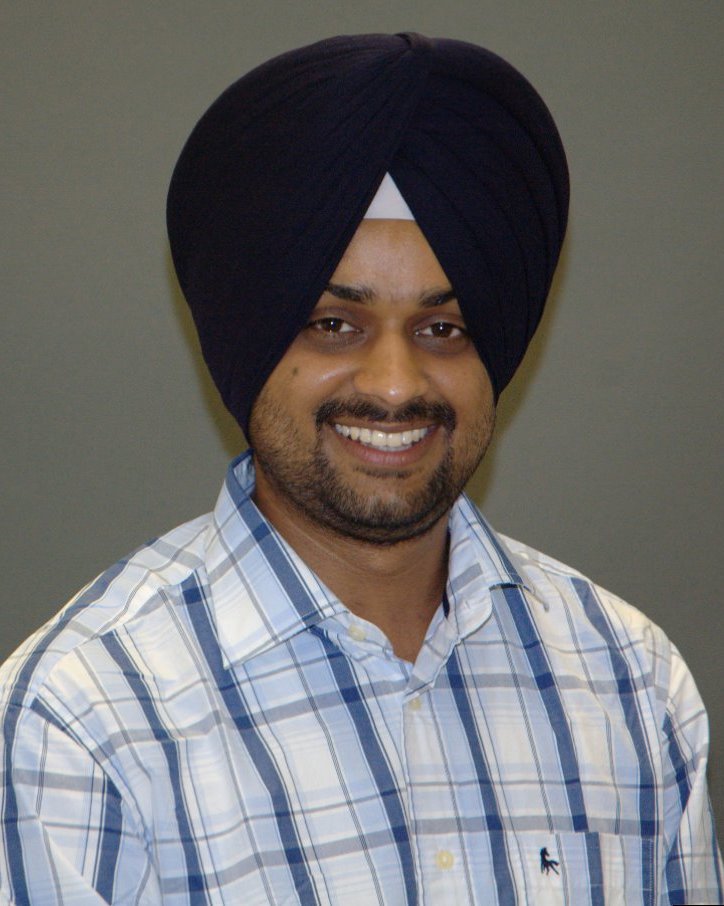}}]{Harpreet S. Dhillon}
(S'11) received the B.Tech. degree in Electronics and Communication Engineering from IIT Guwahati, India, in 2008 and the M.S. in Electrical Engineering from Virginia Tech in 2010. He is currently a Ph.D. student at The University of Texas at Austin, where his research has focused on the modeling and analysis of heterogeneous cellular networks using tools from stochastic geometry, point process theory and spatial statistics. His other research interests include interference channels, multiuser MIMO systems, cognitive radio networks and non-invasive respiration monitoring. He is the recipient of the Microelectronics and Computer Development (MCD) fellowship from UT Austin and was also awarded the Agilent Engineering and Technology Award 2008. He has held summer internships at Alcatel-Lucent Bell Labs in Crawford Hill, NJ, Qualcomm Inc. in San Diego, CA, and Cercom, Politecnico di Torino in Italy.
\end{biography}

\begin{biography}[{\includegraphics[width=1in,height
=1.25in,clip,keepaspectratio]{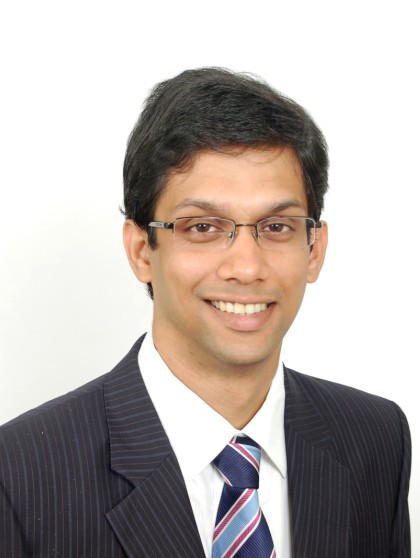}}]{Radha Krishna Ganti}
(S'01, M'10) is an Assistant Professor at the Indian Institute of Technology Madras, Chennai, India. He was a Postdoctoral researcher in the Wireless Networking and Communications Group at UT Austin from 2009-11. He received his B. Tech. and M. Tech. in EE from the Indian Institute of Technology, Madras, and a Masters in Applied Mathematics and a Ph.D. in EE form the University of Notre Dame in 2009. His doctoral work focused on the spatial analysis of interference networks using tools from stochastic geometry. He is a co-author of the monograph {\em Interference in Large Wireless Networks} (NOW Publishers, 2008).
\end{biography}

\begin{biography}[{\includegraphics[width=1in,height
=1.25in,clip,keepaspectratio]{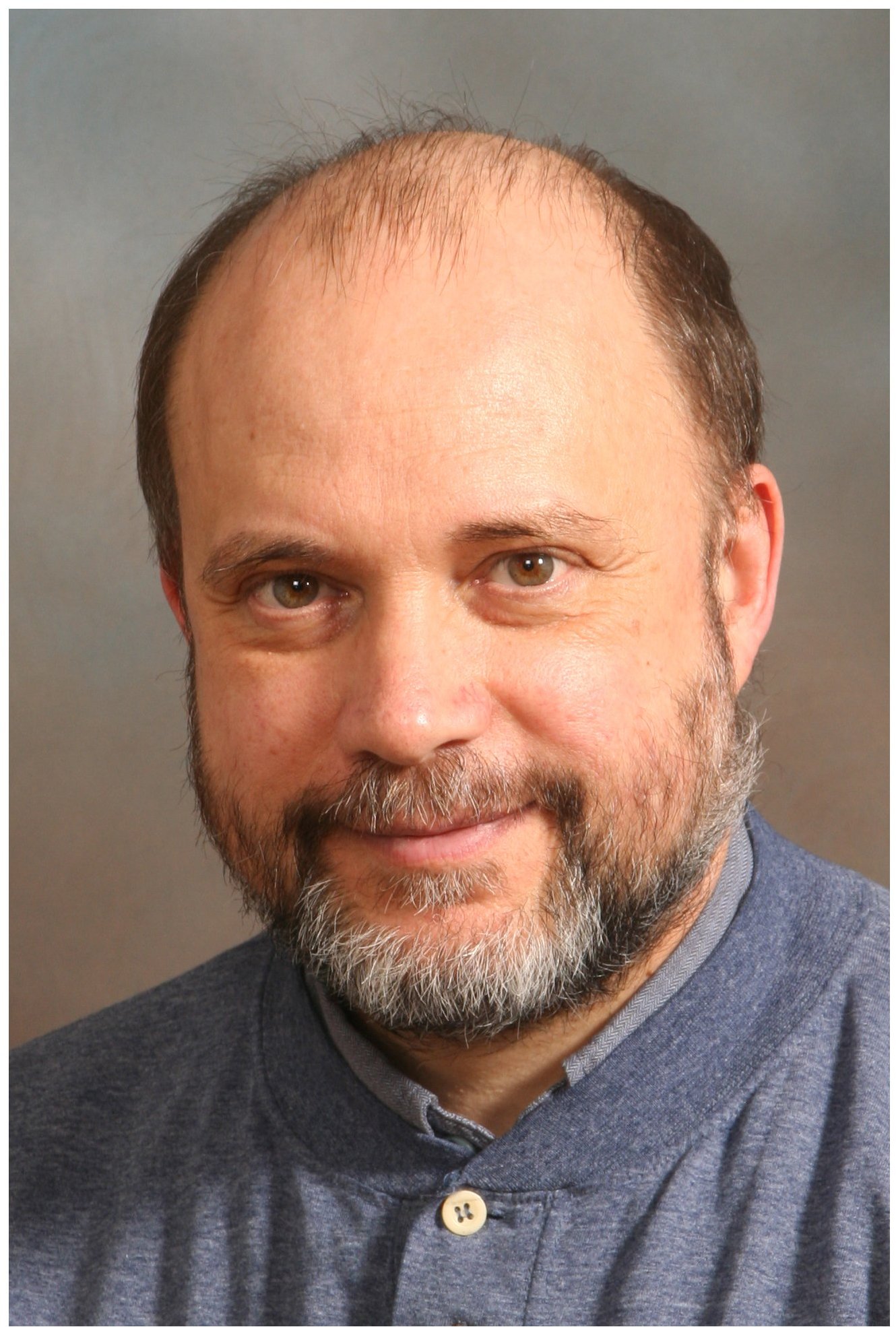}}]{Francois Baccelli}
Francois Baccelli's research interests are in the theory of discrete event dynamical networks and in the modeling and performance evaluation of computer and communication systems. He coauthored more than 100 publications in major international journals and conferences, as well as a 1994 Springer Verlag book on queueing theory, jointly with P. Bremaud, the second edition of which appeared in 2003, and a 1992 Wiley book on the max plus approach to discrete event networks, with G. Cohen, G.J. Olsder and J.P. Quadrat.

He was the head of the Mistral Performance Evaluation research group of INRIA Sophia Antipolis, France, from its creation to 1999. He was/is a partner in several European projects including IMSE (Esprit 2), ALAPEDES (TMR) and EURONGI (Network of Excellence), and was the coordinator of the BRA Qmips European project. He is currently INRIA Directeur de Recherche in the Computer Science Department of Ecole Normale Sup{\'e}rieure in Paris, where he started the TREC group (th{\'e}orie des r{\'e}seaux et communications) in 1999.

His current research interest are focused on the analysis of large wireless networks and the design of scalable and reliable application layers based on the current point to point transport mechanisms, and on the development of new tools for the modeling of protocols with spatial components in wireless networks such as coverage and power control in CDMA and MAC protocols in mobile had hoc networks. Prof. Baccelli was awarded the 2002 France Telecom Prize and got the IBM Award in 2003. He became a member of the French Academy of Sciences in 2005.
\end{biography}

\begin{biography}[{\includegraphics[width=1in,height
=1.25in,clip,keepaspectratio]{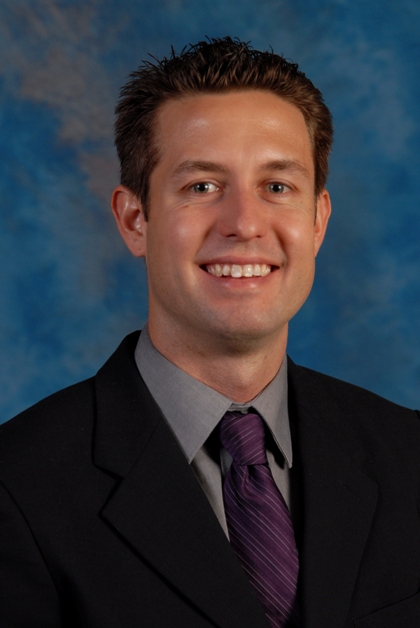}}]{Jeffrey G. Andrews} (S'98, M'02, SM'06) received the B.S. in Engineering with High Distinction from Harvey Mudd College in 1995, and the M.S. and Ph.D. in Electrical Engineering from Stanford University in 1999 and 2002, respectively.  He is an Associate Professor in the Department of Electrical and Computer Engineering at the University of Texas at Austin, where he was the Director of the Wireless Networking and Communications Group (WNCG) from 2008-12. He developed Code Division Multiple Access systems at Qualcomm from 1995-97, and has consulted for entities including the WiMAX Forum, Microsoft, Apple, Clearwire, Palm, Sprint, ADC, and NASA.

Dr. Andrews is co-author of two books, Fundamentals of WiMAX (Prentice-Hall, 2007) and Fundamentals of LTE (Prentice-Hall, 2010), and holds the Earl and Margaret Brasfield Endowed Fellowship in Engineering at UT Austin, where he received the ECE department's first annual High Gain award for excellence in research. He is a Senior Member of the IEEE, served as an associate editor for the IEEE Transactions on Wireless Communications from 2004-08, was the Chair of the 2010 IEEE Communication Theory Workshop, and is the Technical Program co-Chair of ICC 2012 (Comm. Theory Symposium) and Globecom 2014.  He has also been a guest editor for two recent IEEE JSAC special issues on stochastic geometry and femtocell networks.

Dr. Andrews received the National Science Foundation CAREER award in 2007 and has been co-author of five best paper award recipients, two at Globecom (2006 and 2009), Asilomar (2008), the 2010 IEEE Communications Society Best Tutorial Paper Award, and the 2011 Communications Society Heinrich Hertz Prize.  His research interests are in communication theory, information theory, and stochastic geometry applied to wireless cellular and ad hoc networks.
\end{biography}

\end{document}